\documentclass[10pt,conference]{IEEEtran}
\usepackage[pdftex,colorlinks=true,linkcolor=black,citecolor=black,urlcolor=blue]{hyperref}
\usepackage{fullpage}

\usepackage{amsmath,amssymb,amsthm}

\makeatletter
\def\th@plain{%
  \thm@notefont{}
  \itshape 
}
\def\th@definition{%
  \thm@notefont{}
  \normalfont 
}
\makeatother

\usepackage{subcaption}

\usepackage{stmaryrd} 

\usepackage[capitalize,noabbrev,nameinlink]{cleveref} 

\newcommand\xlrsquigarrow{%
  \mathrel{%
    \vcenter{%
      \hbox{%
        \begin{tikzpicture}
          \path[
            draw,
            >={Implies[]},
            <->,
            double distance between line centers=1.5pt,
            decorate,
            decoration={
              zigzag,
              amplitude=0.7pt,
              segment length=3pt,
              pre length=4pt,
              post length=4pt,
            },
          ]   
            (0,0) -- (14pt,0);
        \end{tikzpicture}%
      }%
    }%
  }%
}

\usepackage[utf8]{inputenc}

\usepackage{newunicodechar}
\let\Alpha=A
\let\Beta=B
\let\Epsilon=E
\let\Zeta=Z
\let\Eta=H
\let\Iota=I
\let\Kappa=K
\let\Mu=M
\let\Nu=N
\let\Omicron=O
\let\omicron=o
\let\Rho=P
\let\Tau=T
\let\Chi=X

\newunicodechar{Α}{\ensuremath{\Alpha}}
\newunicodechar{α}{\ensuremath{\alpha}}
\newunicodechar{Β}{\ensuremath{\Beta}}
\newunicodechar{β}{\ensuremath{\beta}}
\newunicodechar{Γ}{\ensuremath{\Gamma}}
\newunicodechar{γ}{\ensuremath{\gamma}}
\newunicodechar{Δ}{\ensuremath{\Delta}}
\newunicodechar{δ}{\ensuremath{\delta}}
\newunicodechar{Ε}{\ensuremath{\Epsilon}}
\newunicodechar{ε}{\ensuremath{\epsilon}}
\newunicodechar{ϵ}{\ensuremath{\varepsilon}}
\newunicodechar{Ζ}{\ensuremath{\Zeta}}
\newunicodechar{ζ}{\ensuremath{\zeta}}
\newunicodechar{Η}{\ensuremath{\Eta}}
\newunicodechar{η}{\ensuremath{\eta}}
\newunicodechar{Θ}{\ensuremath{\Theta}}
\newunicodechar{θ}{\ensuremath{\theta}}
\newunicodechar{ϑ}{\ensuremath{\vartheta}}
\newunicodechar{Ι}{\ensuremath{\Iota}}
\newunicodechar{ι}{\ensuremath{\iota}}
\newunicodechar{Κ}{\ensuremath{\Kappa}}
\newunicodechar{κ}{\ensuremath{\kappa}}
\newunicodechar{Λ}{\ensuremath{\Lambda}}
\newunicodechar{λ}{\ensuremath{\lambda}}
\newunicodechar{Μ}{\ensuremath{\Mu}}
\newunicodechar{μ}{\ensuremath{\mu}}
\newunicodechar{Ν}{\ensuremath{\Nu}}
\newunicodechar{ν}{\ensuremath{\nu}}
\newunicodechar{Ξ}{\ensuremath{\Xi}}
\newunicodechar{ξ}{\ensuremath{\xi}}
\newunicodechar{Ο}{\ensuremath{\Omicron}}
\newunicodechar{ο}{\ensuremath{\omicron}}
\newunicodechar{Π}{\ensuremath{\Pi}}
\newunicodechar{π}{\ensuremath{\pi}}
\newunicodechar{ϖ}{\ensuremath{\varpi}}
\newunicodechar{Ρ}{\ensuremath{\Rho}}
\newunicodechar{ρ}{\ensuremath{\rho}}
\newunicodechar{ϱ}{\ensuremath{\varrho}}
\newunicodechar{Σ}{\ensuremath{\Sigma}}
\newunicodechar{σ}{\ensuremath{\sigma}}
\newunicodechar{ς}{\ensuremath{\varsigma}}
\newunicodechar{Τ}{\ensuremath{\Tau}}
\newunicodechar{τ}{\ensuremath{\tau}}
\newunicodechar{Υ}{\ensuremath{\Upsilon}}
\newunicodechar{υ}{\ensuremath{\upsilon}}
\newunicodechar{Φ}{\ensuremath{\Phi}}
\newunicodechar{φ}{\ensuremath{\varphi}}
\newunicodechar{ϕ}{\ensuremath{\phi}}
\newunicodechar{Χ}{\ensuremath{\Chi}}
\newunicodechar{χ}{\ensuremath{\chi}}
\newunicodechar{Ψ}{\ensuremath{\Psi}}
\newunicodechar{ψ}{\ensuremath{\psi}}
\newunicodechar{Ω}{\ensuremath{\Omega}}
\newunicodechar{ω}{\ensuremath{\omega}}

\newunicodechar{ℂ}{\ensuremath{\mathbb{C}}}
\newunicodechar{ℤ}{\ensuremath{\mathbb{Z}}}
\newunicodechar{ℝ}{\ensuremath{\mathbb{R}}}
\newunicodechar{ℕ}{\ensuremath{\mathbb{N}}}
\newunicodechar{ℚ}{\ensuremath{\mathbb{Q}}}

\newunicodechar{ℰ}{\ensuremath{\mathcal{E}}}
\newunicodechar{ℱ}{\ensuremath{\mathcal{F}}}
\newunicodechar{ℋ}{\ensuremath{\mathcal{H}}}
\newunicodechar{ℑ}{\ensuremath{\mathcal{I}}}
\newunicodechar{ℒ}{\ensuremath{\mathcal{L}}}
\newunicodechar{ℳ}{\ensuremath{\mathcal{M}}}
\newunicodechar{ℛ}{\ensuremath{\mathcal{R}}}
\newunicodechar{ℜ}{\ensuremath{\mathcal{R}}}
\newunicodechar{ℓ}{\ensuremath{\ell}}

\newunicodechar{∼}{\ensuremath{\sim}}
\newunicodechar{≈}{\ensuremath{\approx}}
\newunicodechar{≋}{\ensuremath{\bumpeq}}
\newunicodechar{≅}{\ensuremath{\cong}}
\newunicodechar{≡}{\ensuremath{\equiv}}
\newunicodechar{≤}{\ensuremath{\le}}
\newunicodechar{≥}{\ensuremath{\ge}}
\newunicodechar{≲}{\ensuremath{\lesssim}}
\newunicodechar{≠}{\ensuremath{\neq}}
\newunicodechar{≔}{\ensuremath{\coloneqq}}
\newunicodechar{⋆}{\ensuremath{\applysymbol}}
\newunicodechar{ƛ}{\ensuremath{\linearlambda}}

\newunicodechar{±}{\ensuremath{\pm}}
\newunicodechar{∓}{\ensuremath{\pm}}
\newunicodechar{×}{\ensuremath{\times}}

\newunicodechar{→}{\ensuremath{\to}}
\newunicodechar{←}{\ensuremath{\leftarrow}}
\newunicodechar{⇒}{\ensuremath{\Rightarrow}}
\newunicodechar{↦}{\ensuremath{\mapsto}}
\newunicodechar{↝}{\ensuremath{\leadsto}}
\newunicodechar{⸖}{\ensuremath{\mathrel{>!}}}
\newunicodechar{⇆}{\ensuremath{\xlrsquigarrow}}
\newunicodechar{↔}{\ensuremath{\leftrightarrow}}
\newunicodechar{↭}{\ensuremath{\leftrightsquigarrow}}
\newunicodechar{↣}{\ensuremath{\rightarrowtail}}

\newunicodechar{∀}{\ensuremath{\forall}}
\newunicodechar{∃}{\ensuremath{\exists}}
\newunicodechar{∨}{\ensuremath{\vee}}
\newunicodechar{∧}{\ensuremath{\wedge}}
\newunicodechar{⊢}{\ensuremath{\vdash}}
\newunicodechar{⊣}{\ensuremath{\dashv}}
\newunicodechar{⊤}{\ensuremath{\top}}
\newunicodechar{⊥}{\ensuremath{\bot}}
\newunicodechar{¬}{\ensuremath{\neg}}

\newunicodechar{⊸}{\ensuremath{\multimap}}
\newunicodechar{⊗}{\ensuremath{\otimes}}
\newunicodechar{⊕}{\ensuremath{\oplus}}
\newunicodechar{¡}{\ensuremath{\oc}}
\newunicodechar{⨂}{\ensuremath{\bigotimes}}
\newunicodechar{⨁}{\ensuremath{\bigoplus}}

\newunicodechar{∈}{\ensuremath{\in}}
\newunicodechar{∉}{\ensuremath{\not\in}}
\newunicodechar{⊆}{\ensuremath{\subseteq}}
\newunicodechar{∪}{\ensuremath{\cup}}
\newunicodechar{⋓}{\ensuremath{\Cup}}
\newunicodechar{∅}{\ensuremath{\emptyset}}

\newunicodechar{〈}{\ensuremath{\langle}}
\newunicodechar{⟨}{\ensuremath{\langle}}
\newunicodechar{⟩}{\ensuremath{\rangle}}
\newunicodechar{〉}{\ensuremath{\rangle}}
\newunicodechar{⟦}{\ensuremath{\llbracket}}
\newunicodechar{⟧}{\ensuremath{\rrbracket}}

\newunicodechar{₀}{\ensuremath{_0}}
\newunicodechar{₁}{\ensuremath{_1}}
\newunicodechar{₂}{\ensuremath{_2}}
\newunicodechar{₃}{\ensuremath{_3}}
\newunicodechar{₄}{\ensuremath{_4}}
\newunicodechar{₅}{\ensuremath{_5}}
\newunicodechar{₆}{\ensuremath{_6}}
\newunicodechar{₇}{\ensuremath{_7}}
\newunicodechar{₈}{\ensuremath{_8}}
\newunicodechar{₉}{\ensuremath{_9}}
\newunicodechar{ₐ}{\ensuremath{_a}}
\newunicodechar{ₑ}{\ensuremath{_e}}
\newunicodechar{ₕ}{\ensuremath{_h}}
\newunicodechar{ᵢ}{\ensuremath{_i}}
\newunicodechar{ⱼ}{\ensuremath{_j}}
\newunicodechar{ₖ}{\ensuremath{_k}}
\newunicodechar{ₗ}{\ensuremath{_l}}
\newunicodechar{ₘ}{\ensuremath{_m}}
\newunicodechar{ₙ}{\ensuremath{_n}}
\newunicodechar{ₒ}{\ensuremath{_o}}
\newunicodechar{ₚ}{\ensuremath{_p}}
\newunicodechar{ᵣ}{\ensuremath{_r}}
\newunicodechar{ₛ}{\ensuremath{_s}}
\newunicodechar{ₜ}{\ensuremath{_t}}
\newunicodechar{ᵤ}{\ensuremath{_u}}
\newunicodechar{ₓ}{\ensuremath{_x}}

\newunicodechar{⁰}{\ensuremath{^0}}
\newunicodechar{¹}{\ensuremath{^1}}
\newunicodechar{²}{\ensuremath{^2}}
\newunicodechar{³}{\ensuremath{^3}}
\newunicodechar{⁴}{\ensuremath{^4}}
\newunicodechar{⁵}{\ensuremath{^5}}
\newunicodechar{⁶}{\ensuremath{^6}}
\newunicodechar{⁷}{\ensuremath{^7}}
\newunicodechar{⁸}{\ensuremath{^8}}
\newunicodechar{⁹}{\ensuremath{^9}}
\newunicodechar{ᵃ}{\ensuremath{^a}}
\newunicodechar{ᵇ}{\ensuremath{^b}}
\newunicodechar{ᶜ}{\ensuremath{^c}}
\newunicodechar{ᵈ}{\ensuremath{^d}}
\newunicodechar{ᵉ}{\ensuremath{^e}}
\newunicodechar{ᶠ}{\ensuremath{^f}}
\newunicodechar{ᵍ}{\ensuremath{^g}}
\newunicodechar{ʰ}{\ensuremath{^h}}
\newunicodechar{ⁱ}{\ensuremath{^i}}
\newunicodechar{ʲ}{\ensuremath{^j}}
\newunicodechar{ᵏ}{\ensuremath{^k}}
\newunicodechar{ᵐ}{\ensuremath{^m}}
\newunicodechar{ⁿ}{\ensuremath{^n}}
\newunicodechar{ᵒ}{\ensuremath{^o}}
\newunicodechar{ᵖ}{\ensuremath{^p}}
\newunicodechar{ˢ}{\ensuremath{^s}}
\newunicodechar{ᵗ}{\ensuremath{^t}}
\newunicodechar{ᵘ}{\ensuremath{^u}}
\newunicodechar{ʷ}{\ensuremath{^v}}
\newunicodechar{ˣ}{\ensuremath{^x}}
\newunicodechar{ʸ}{\ensuremath{^y}}
\newunicodechar{ᶻ}{\ensuremath{^z}}

\newunicodechar{•}{\ensuremath{\bullet}}
\newunicodechar{∙}{\ensuremath{\bullet}}
\newunicodechar{·}{\ensuremath{\cdot}}
\newunicodechar{⋯}{\ensuremath{\cdots}}
\newunicodechar{…}{\ensuremath{\ldots}}
\newunicodechar{∷}{\ensuremath{~\mathrel{:\!\!\!:}~}}

\newunicodechar{∣}{\ensuremath{\mid}}
\newunicodechar{∥}{\ensuremath{ {\parallel} }}

\newunicodechar{□}{\ensuremath{\square}}

\newunicodechar{∗}{\ensuremath{\ast}}

\newunicodechar{∘}{\ensuremath{\circ}}
\newunicodechar{†}{\ensuremath{\dagger}}
\newunicodechar{♯}{\ensuremath{\mathrel{\#}}}

\newunicodechar{∞}{\ensuremath{\infty}}
\newunicodechar{£}{\ensuremath{\mathrel{\$}}}

\newunicodechar{⧺}{\ensuremath{\doubleplus}}

\usepackage{aliascnt}

\crefformat{equation}{#2Equation~#1#3}

\usepackage{thmtools,thm-restate}

\newtheorem{theorem}{Theorem}

\newtheorem{corollary}[theorem]{Corollary}
\newtheorem{lemma}[theorem]{Lemma}

\theoremstyle{definition} 
\newtheorem{definition}{Definition}

\usepackage{url}

\usepackage{breakurl}

\usepackage{mathpartir}

\usepackage{booktabs}

\usepackage{xcolor}
\definecolor{ltblue}{rgb}{0,0.4,0.4}
\definecolor{dkblue}{rgb}{0,0.1,0.6}
\definecolor{dkgreen}{rgb}{0,0.35,0}
\definecolor{dkviolet}{rgb}{0.3,0,0.5}
\definecolor{dkred}{rgb}{0.5,0,0}

\usepackage{listings}
\input{lstcoq.sty}

\usepackage{tikz} 
\usetikzlibrary{cd}
\usetikzlibrary{positioning}
\usetikzlibrary{%
  arrows,%
  shapes.misc,
  shapes.arrows,%
  shapes.callouts,
  chains,%
  matrix,%
  positioning,
  scopes,%
  decorations.pathmorphing,
  decorations.text,
  shadows%
}

\usepackage{textcomp} 

\usepackage{xparse}
\NewDocumentCommand{\optionalParens}{s m m}{
    \IfBooleanTF{#2}{(#3)}{\IfBooleanTF{#1}{~#3}{#3}}
}
\NewDocumentCommand{\apply}{m s O{} m}{
    #1#3 \optionalParens*{#2}{#4}
}

\NewDocumentCommand{\applytwo}{m O{} s m s m}{ 
   #1#2~\optionalParens{#3}{#4}~\optionalParens{#5}{#6}
}

\usepackage{xspace}

\usepackage{cmll} 

\usepackage{braket}

\usepackage{comment}
\usepackage{draft}

\makeatletter
\NewStdDefinerWrappers{DraftFormat}{%
 \ifdraft
   \expandafter\@firstoftwo
 \else
   \expandafter\@secondoftwo
 \fi
 {#5}{\@firstofone}%
}

\usepackage[normalem]{ulem}
\newcommand*{\coloruwave}[1]{%
 \bgroup
 \markoverwith{\lower3.5\p@\hbox{\sixly\textcolor{#1}{\char58}}}%
 \ULon}
\NewDraftFormat{\awk}{\coloruwave{red}}
\makeatother

\usepackage{enumerate}
\usepackage[shortlabels,inline]{enumitem}

\drafttrue

\newnote[JP]{jennifer}{red}
\newnote[PB]{peter}{blue}
\newnote{fixme}{green}
\newnote{todo}{green}
\newnote[Cite:]{tocite}{green}

\makeatletter
\newcommand{\xRightarrow}[2][]{\ext@arrow 0359\Rightarrowfill@{#1}{#2}}
\makeatother

\NewDocumentCommand{\async}{m m m m O{} m}{\langle #1, #2 \rangle \vdash #3 \downarrow #4 \mapsto^{#5} #6}
\newcommand{\mg}[3]{#1 \vdash #2 \downarrow #3}

\newcommand{\synceval}[3]{\coqmath{sync}(#1,#2,#3)}
\NewDocumentCommand{\synceven}{mmo}{
  \coqmath{sync_even}(\coqmath{#1},#2 \IfValueT{#3}{,#3})
  }
\NewDocumentCommand{\syncodd}{mmo}{
  \coqmath{sync_odd}(\coqmath{#1},#2 \IfValueT{#3}{,#3})
  }
\newcommand{\nextstate}[3]{\coqmath{next_state}(#1,#2,#3)}

\newcommand{\tempty}{\coqmath{t_empty}}
\newcommand{\tcons}[2]{#2 \triangleright #1}

\newcommand{\fire}[3]{\coqmath{fire}(#1,#2,#3)}
\newcommand{\transparent}[2]{\coqmath{latch_transparency}(#1,#2)}
\newcommand{\numevents}{\applytwo{\coqmath{num_events}}}

\newcommand{\SRC}{\textcolor{gray}{\scriptsize{\textsc{SRC}}}}
\newcommand{\SNK}{\textcolor{gray}{\scriptsize{\textsc{SNK}}}}

\usepackage[numbers,sort&compress]{natbib}

\usepackage{caption}
\title{Formal Verification of Flow Equivalence in Desynchronized Designs}

\author{
\IEEEauthorblockN{Jennifer Paykin,\IEEEauthorrefmark{1} Brian Huffman,\IEEEauthorrefmark{1} Daniel M. Zimmerman,\IEEEauthorrefmark{1}
Peter A. Beerel\IEEEauthorrefmark{1}\IEEEauthorrefmark{2}
}
\IEEEauthorblockA{\IEEEauthorrefmark{1}Galois, Inc,
Portland, OR, USA \\
\IEEEauthorrefmark{2}Ming Hsieh Department of Electrical and Computer Engineering, USC, 
Los Angeles, CA, USA \\
\texttt{\{jpaykin,huffman,dmz\}@galois.com}, \texttt{pabeerel@usc.edu} }
}

\begin{document}

\maketitle

\begin{abstract}
Seminal work by \citeauthor*{Cortadella2006} includes a hand-written proof that a particular handshaking protocol preserves \emph{flow equivalence}, a notion of equivalence between synchronous latch-based specifications and their desynchronized bundled-data asynchronous implementations. In this work we identify a counterexample to \citeauthor{Cortadella2006}'s proof illustrating how their protocol can in fact lead to a violation of flow equivalence. However, two of the less concurrent protocols identified in their paper do preserve flow equivalence. To verify this fact, we formalize flow equivalence in the Coq proof assistant and provide mechanized, machine-checkable proofs of our results.
\end{abstract}

\section{Introduction}

\emph{Desynchronization} is a popular strategy for designing bundled-data (BD) asynchronous latch-based designs from a synchronous RTL specification. In a desynchronized design, master/slave latches are driven by local clocks controlled by specific classes of asynchronous controllers~\cite{Cortadella2006}. 
Because BD latch-based designs are susceptible to subtle timing assumptions, it is important to ensure that desynchronization preserves the functional behavior of the circuit. \citet{LeGuernic2003} propose \emph{flow equivalence} as a property characterizing when two circuits, either synchronous or asynchronous, have the same functional behavior. \citet{Cortadella2006} provide a hand-written proof that a highly concurrent desynchronization controller preserves flow equivalence, and use that result to claim that two less concurrent controllers, called \emph{rise-decoupled} and \emph{fall-decoupled}, also preserve flow equivalence.

We identify a subtle flaw in \citeauthor{Cortadella2006}'s proof and provide a counterexample showing that the proposed desynchronization controller fails to guarantee flow equivalence. Subtle errors in hand-written proofs are easy to overlook, as evidenced by the wide acceptance of \citeauthor{Cortadella2006}'s results.

To ensure that our results are correct, we formalize flow equivalence in the higher-order interactive theorem prover Coq~\cite{coq2019}, which codifies high-level mathematics (such as induction, case analysis, functions, and relations) in a machine-checkable system. In particular, we show that the core ideas of \citeauthor{Cortadella2006}'s work are sound by adapting their flow equivalence proof to the less concurrent fall- and rise-decoupled controllers.

This work illustrates the benefits of applying formal methods to asynchronous design. 
Other formal techniques in the literature range from verification of hazard-freedom (e.g., \cite{Nelson07}) and deadlock-freedom (e.g., \cite{Verbeek2013}), 
to more general notions of equivalence 
between gate level implementations, handshaking-level specifications \cite{Cunning04,Park2016,Longfield13,Sakib2019}, and abstract asynchronous communication primitives~\cite{Saifhash2015}. 
Several researchers focus on verifying general properties of asynchronous designs using model checkers \cite{Bui12,Borrione03, Bouzafour2018} and proof assistants, including ACL2~\cite{Peng2019,Chau2019}.

This paper makes several key contributions. \cref{sec:counter} describes a counterexample to \citeauthor{Cortadella2006}'s hand-written proof in the form of a timing diagram that satisfies the desynchronization protocol but violates flow equivalence, motivating our formal analysis. \cref{sec:back} formalizes the definition of flow equivalence as well as the marked graph specification of the proposed controllers, and \cref{sec:proof} presents Coq proofs that the rise-decoupled and fall-decoupled controllers guarantee flow equivalence.\footnotemark~ Finally, \cref{sec:summary} describes some directions for future research, including extending the formalization to \citeauthor{Cortadella2006}'s liveness results.

\footnotetext{All Coq definitions and proofs can be found in our formalization available at \texttt{\url{https://github.com/GaloisInc/Coq-Flow-Equivalence}}.}

\section{Counterexample}
\label{sec:counter}

\citeauthor{Cortadella2006}'s desynchronization protocol~\citep{Cortadella2006} defines a set of requirements for asynchronous controllers of a latch-based circuit. In this section we demonstrate that the desynchronization protocol fails to preserve flow equivalence in general; we give a circuit and a set of clock traces that is valid according to the protocol, but produces different outputs than the synchronous version of the same circuit.


Consider a three-stage linear pipeline with latches $A$, $B$, and $C$. Latches $A$ and $C$ are \emph{even}, and have valid initial values $a_0$ and $c_0$; latch $B$ is \emph{odd}, and is uninitialized. When driven synchronously, values for cycle $n$ are produced for odd latches first, followed by even latches. For $n > 0$ we have $b_n = F_B(a_{n-1})$ and $c_n = F_C(b_{n})$, where $F_B$ and $F_C$ represent the combinational logic before each latch (\cref{fig:synchronous}).

\begin{figure}
    \centering
    \includegraphics[scale=0.85]{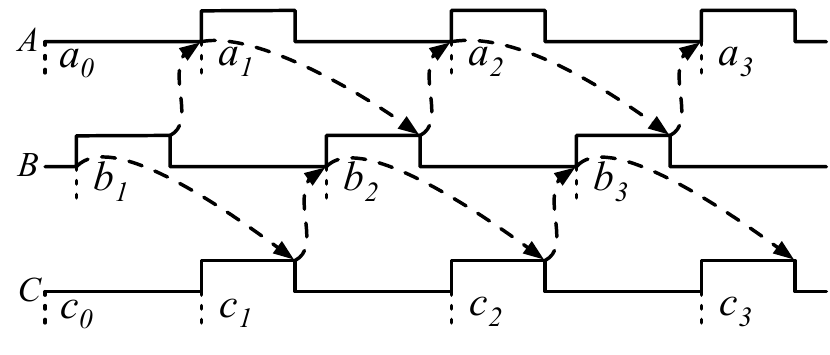}
    \caption{Three-stage pipeline with synchronous timing. Arrows indicate timing constraints of the desynchronized protocol.}
    \label{fig:synchronous}
\end{figure}

The dashed arrows in \cref{fig:synchronous} indicate the timing constraints of the desynchronized protocol, derived from \citeauthor{Cortadella2006}'s marked graph specification; asynchronous controllers are allowed to shift the clock transitions arbitrarily forward and backward in time, as long as those constraints are preserved.\footnotemark~ We construct a counterexample by delaying the first transitions for latch $A$ as long as possible (\cref{fig:counterexample}).

\footnotetext{\Cref{fig:synchronous} and \citeauthor{Cortadella2006}'s proof both assume that all latches start opaque, contrary to the convention used later in this paper (e.g., \cref{fig:protocols}), where odd latches start transparent. It is straightforward to show that these initial states are mutually reachable from each other in the marked graph protocols.}

\begin{figure}
    \centering
    \includegraphics[scale=0.85]{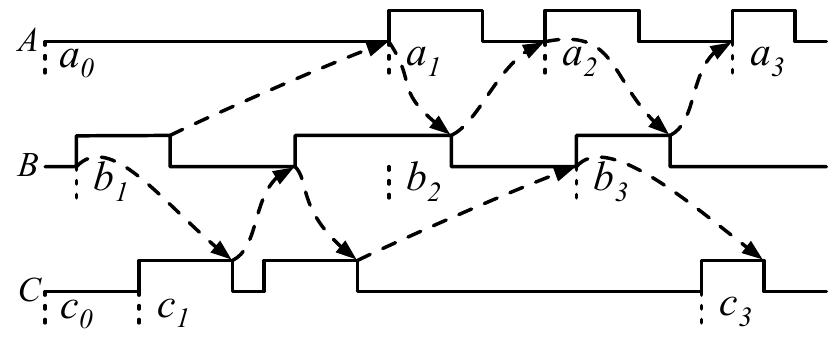}
    \caption{Asynchronous circuit timing that obeys the desynchronized protocol, but violates flow equivalence.}
    \label{fig:counterexample}
\end{figure}

Two circuits are flow-equivalent if the sequences of values stored at their corresponding latches are identical. The trace of values on latch $C$ in \cref{fig:counterexample} violates flow equivalence: instead of $[c_0,c_1,c_2,c_3]$, the asynchronous circuit latches $[c_0,c_1,c_1,c_3]$. Crucially, the second clock pulse of $C$ ends before the first pulse of $A$ begins. At this point $C$ should have latched $c_2$, but $c_2$ has a transitive data dependency on $a_1$, which is not yet available.

The error in \citeauthor{Cortadella2006}'s proof arises from a subtle misuse of its induction hypothesis. The proof states that, immediately before the second fall of an even latch's clock ($C$), its predecessor's clock ($B$) must have risen twice, and therefore the induction hypothesis implies that $B$ has the value $b_2$. This is not guaranteed; the induction hypothesis only applies after $B-$ occurs. It appears that the proof assumes that $B$'s new value is available as soon as $B+$ occurs, which, as our counterexample shows, is not necessarily the case.\footnotemark

\footnotetext{In fact, the invariant that a new value is always available as soon as a latch becomes transparent \emph{is} true for the less concurrent fall-decoupled controller, but not for the rise-decoupled or desynchronization controllers, as detailed in \cref{sec:proof}.}

The proof contains a similar statement about even predecessors of odd latches, and it is easy to construct a counterexample to illustrate that case. More complex configurations, including cycles, also give rise to counterexamples.

\section{Formalizing Flow Equivalence}
\label{sec:back}

In this section we formalize the definition of flow equivalence in the Coq proof assistant in a way that is amenable to verification later on.
Each Coq definition automatically generates an induction principle reflecting the definition's structure, so different definitions can make proofs easier or harder. Though there are many possible definitions of flow equivalence, the one we present here yields concise and elegant proofs.

Coq is a dependently-typed functional programming language and interactive proof assistant that has been used  both to prove significant mathematical theorems~\cite{Gonthier2008} and to verify the correctness of realistic software systems~\cite{Leroy2012}. Proof engineers interactively write proof scripts to prove theorems in Coq, from which the system produces a proof term that the Coq kernel can check for validity. Because the kernel is small and trustworthy, the proof scripts used to generate proof terms can be arbitrarily complex and the user can still have confidence that their proof is correct.

We develop the definition of flow equivalence in several stages, illustrating Coq syntax along the way. The resulting formalization could almost certainly be replicated in other  provers such as ACL2, Isabelle, or $F^\ast$, but Coq's interactivity, scriptability, and rich inductive definitions were particularly useful in this development.

\subsection{Synchronous execution semantics}

We first define the synchronous execution of a circuit. 
Like \citeauthor{Cortadella2006}, we model circuits with master-slave latches obtained by decoupling D flip-flops; every latch is either \emph{even} (master) or \emph{odd} (slave), and neighboring latches must have opposite parities.

Let \coq{even} be a set of even latches and let \coq{odd} be a set of odd latches. We define the type \coq{latch} of all latches to be made up of both even and odd latches---mathematically, the disjoint union of \coq{odd} and \coq{even}.
\begin{Coq}
Inductive latch := Odd : odd -> latch
                 | Even : even -> latch.
\end{Coq}
In the Coq code, the keyword \coq{Inductive} indicates that \coq{latch} is defined by its constructors, \coq{Odd} and \coq{Even}. \coq{Odd} is a function of type \coq{odd -> latch}, so takes an argument drawn from the set \coq{odd}, and produces a \coq{latch}. 

A \emph{state} relative to a type $A$ is a function from $A$ to values, which are either numbers or undefined.
\begin{Coq}
Inductive value := Num : nat -> value | X  : value.
Definition state (A : Type) := A -> value.
\end{Coq}

\begin{definition}[Circuits]
A \emph{circuit} consists of lists of neighboring even-odd and odd-even latches, as well as, for each latch, a function that computes its value from the values of its left neighbors. Like \citeauthor{Cortadella2006}, we model only closed circuits and assume that combinational blocks and latches have zero delay. Any open circuit can be combined with a model of the environment to obtain a closed circuit.
\begin{Coq}
Record circuit :=
  { even_odd_neighbors : list (even * odd)
  ; odd_even_neighbors : list (odd * even)
  ; next_state_e (E : even) :
    state {O : odd & In (O,E) odd_even_neighbors} -> value
  ; next_state_o (O : odd) :
    state {E : even & In (E,O) even_odd_neighbors} -> value}.
\end{Coq}
\end{definition}
 The \coq{Record} syntax defines a data type with a collection of functions out of that data type. The \coq{circuit} record has four such functions: \coq{eo_nbrs} and \coq{oe_nbrs} are functions from a \coq{circuit} to lists of its even-odd and odd-even neighbors, respectively; \coq{next_e} is a function from a \coq{circuit}, an \coq{even} latch $E$, and a \coq{state} of the \coq{odd} left neighbors of $E$ to the next state of $E$; \coq{next_o} is similar to \coq{next_e}, but for \coq{odd} latches.
 
 In the \coq{Record}, the notation \coq|{x : A & P(x)}| restricts the type \coq{A} to those elements satisfying the predicate \coq{P}. For example, the \coq{next_e} function takes as input a state restricted to the odd latches \coq{O} satisfying the predicate \coq{In (O, E) oe_nbrs}, where \coq{In a ls} indicates that the element \coq{a} occurs in the list \coq{ls}.

When a latch $l$ can be either even or odd, we write $\nextstate{c}{st}{l}$ for the corresponding next-state function. 

\begin{definition}[Synchronous Execution]
The \emph{synchronous execution} of a circuit $c$ is a function from the initial state $st_0$ of the circuit to the state at its $n$th clock cycle. In this paper, odd latches update first each clock cycle, followed by even latches.\footnote{Without loss of generality, this ordering could be reversed so that even latches update first; the initial markings of the protocols in \cref{fig:protocols} would need to be updated to reflect that convention. } That is:\footnote{Note that this definition is well-founded: the values of even latches at time $n>0$ depend on the values of their odd predecessors at time $n$, which themselves depend on the values of their even predecessors at time $n-1$.}
\begin{align*}
    &\synceval{\coqmath{c}}{st_0}{n}(l) = \\
    &
        \begin{cases}
            st_0(l) &\text{if}~n = 0 \wedge l ~\text{is even} \\
            X &\text{if}~n=0 \wedge l~\text{is odd} \\
            \nextstate{\coqmath{c}}{\synceval{\coqmath{c}}{st_0}{n}}{l} &\text{if}~n > 0 \wedge l~\text{is even} \\
            \nextstate{\coqmath{c}}{\synceval{\coqmath{c}}{st_0}{n-1}}{l} &\text{if}~ n > 0 \wedge l~\text{is odd}
        \end{cases}
\end{align*}

\end{definition}

\subsection{Asynchronous execution semantics}

Next, we describe how such a circuit can be executed \emph{asynchronously} by replacing the shared clock with a series of local clocks controlled by bundled data controllers. This section describes the asynchronous execution semantics with respect to sequences of rising and falling transitions of the local clocks. The allowable sequences are constrained by controller specifications described in \cref{sec:mg}.  

\begin{definition}[Events and Traces]
An \emph{event} $e$ is the rise (written $l+$) or fall (written $l-$) of a latch $l$. A \emph{trace} is a list of events---rises and falls of a latch's clock. When a latch's clock rises, the latch becomes transparent; when it falls, the latch becomes opaque.
\begin{Coq}
Inductive event := Rise : latch -> event
                 | Fall : latch -> event.
Definition trace := tail_list event.
\end{Coq}
\end{definition}
A \coq{tail_list} is either empty, denoted $\tempty$, or is an element $a$ appended onto another \coq{tail_list} $t$, denoted $\tcons{a}{t}$.

\begin{definition}[Transparency]
A \emph{transparency} characterizes whether a latch is transparent or opaque at a particular point.
\begin{Coq}
Inductive transparency := Transparent | Opaque.
\end{Coq}
The function \coq{latch_transparency} computes the transparency of a latch $l$ after executing a trace $t$.
\begin{align*} &\transparent{t}{l}= \\
  &\begin{cases}
    \coqmath{Transparent} &\text{if}~ t=\tempty \wedge l~\text{is odd} \\
    \coqmath{Opaque} &\text{if}~t=\tempty \wedge l~\text{is even} \\
    \coqmath{Transparent} &\text{if}~t = \tcons{l+}{t'} \\
    \coqmath{Opaque} &\text{if}~t = \tcons{l-}{t'} \\
    \transparent{t'}{l} &\text{if}~t = \tcons{e}{t'}~\text{otherwise}
  \end{cases}
\end{align*}
\end{definition}

\begin{definition}[Asynchronous Execution]
\label{def:async-execution}
  Given an initial state $st_0$, the \emph{asynchronous execution} of a trace $t$ on a circuit $c$ is defined by a 5-ary relation, for which we introduce the syntax ${\async{c}{st_0}{t}{l}{v}}$. Here, $l$ is a latch and $v$ is the value of $l$ after executing $t$. When this relation holds of a 5-tuple, we say that \emph{$l$ evaluates to $v$ by means of $t$}:
    \begin{itemize}[leftmargin=*]
    
      \item If $l$ is transparent in $t$ and for all left neighbors $l'$ of $l$ it is the case that $\async{c}{st_0}{t}{l'}{st(l')}$, then
     $
        \async{c}{st_0}{t}{l}{\nextstate{c}{st}{l}}.
     $
     
      \item If $l$ is opaque in $t$ and $t$ is the empty list, then ${\async{c}{st_0}{\tempty}{l}{st_0(l)}}$.
      
    \item If $l$ is opaque in $t$ and $t = \tcons{e}{t'}$ such that $e \neq l-$, then $\async{c}{st_0}{t'}{l}{v}$ implies $\async{c}{st_0}{\tcons{e}{t'}}{l}{v}$.
      
     \item Finally, if  $t =\tcons{l-}{t'}$ and for all left neighbors $l'$ of $l$ it is the case that $\async{c}{st_0}{t'}{l'}{st(l')}$, then
     $
       \async{c}{st_0}{\tcons{l-}{t'}}{l}{\nextstate{c}{st}{l}}
     $.

  \end{itemize}
\end{definition}

\Cref{def:async-execution} has four cases, depending on both the transparency of $l$ in $t$ and the structure of $t$. This choice of decomposition gives rise to a strong induction principle based on those four cases, which we will use in \cref{sec:proof} to prove flow equivalence. 

\Cref{fig:async} shows the Coq definition of this relation, extended with an additional parameter $O$ corresponding to the transparency of $l$ in $t$. The additional parameter makes the Coq relation easier to work with, as we often need to consider only opaque or only transparent latches. 
  Note that function application in Coq is written as the function name (e.g. \coq{latch_transparency}) next to its arguments (e.g. \coq{t}, \coq{l}) with whitespace separators. 
  
   Each constructor in the Coq definition is a logical formula that indicates when the relation is satisfied.
For example, when the trace is empty the circuit is in its initial state: even latches are opaque and they have their corresponding initial values. This property is written $\async{c}{st_0}{\tempty}{E}[\coqscriptmath{Opaque}]{st_0(E)}$, and the Coq proof term \coq{async_nil E}  is evidence that it holds.

\begin{figure*}
    \begin{center}
    \begin{minipage}[c]{0.8\linewidth}
    \begin{Coq}
Inductive async (c : circuit) (st0 : state latch)
                : trace -> latch -> transparency -> value -> Prop :=
| async_transparent : forall l t st v, transparent t l = Transparent ->
    (forall l', neighbor c l' l -> ⟨c,st0⟩⊢ t ↓ l' ↦{transparent t l'} st l') ->
    ⟨c,st0⟩⊢ t ↓ l ↦{Transparent} next_state c st l
    
| async_nil : forall E,
    ⟨c,st0⟩ ⊢ t_empty ↓ Even E ↦{Opaque} st0 (Even E)

| async_opaque : forall l e t' v, transparent (t' ▷ e) l = Opaque ->
    e <> Fall l ->
    ⟨c,st0⟩⊢ t' ↓ l ↦{Opaque} v ->
    ⟨c,st0⟩⊢ t' ▷ e ↓ l ↦{Opaque} v

| async_opaque_fall : forall l e t' v st,
    (forall l', neighbor c l' l -> ⟨c,st0⟩⊢ t' ↓ l' ↦{transparent t' l'} st l') ->
    ⟨c,st0⟩⊢ t' ▷ Fall l ↓ l ↦{Opaque} next_state c st l

where "⟨ c , st ⟩⊢ t ↓ l ↦{ O } v" := (async c st t l O v).
    \end{Coq}
    \end{minipage}
    \end{center}
    \vspace{-18pt}
    \caption{Coq definition of asynchronous execution of a trace.}
    \label{fig:async}
\end{figure*}

\begin{figure*}
    \centering
    \begin{subfigure}{0.3\textwidth}
    \begin{center}
        \includegraphics[scale=0.54]{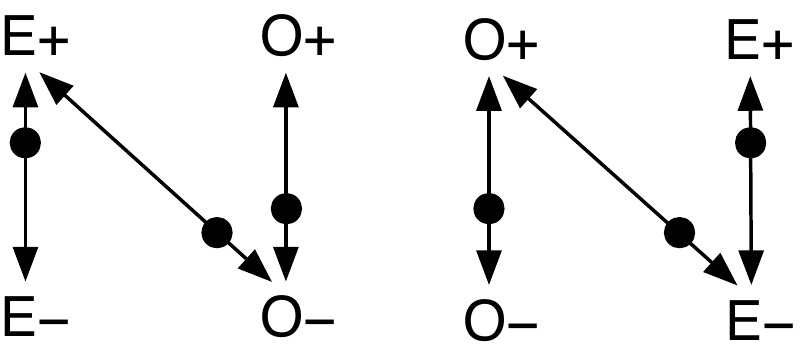}
    \end{center}
    \caption{Desynchronization Protocol}
    \label{fig:desync}
    \end{subfigure}
    \qquad
    \begin{subfigure}{0.3\textwidth}
    \begin{center}
        \includegraphics[scale=0.54]{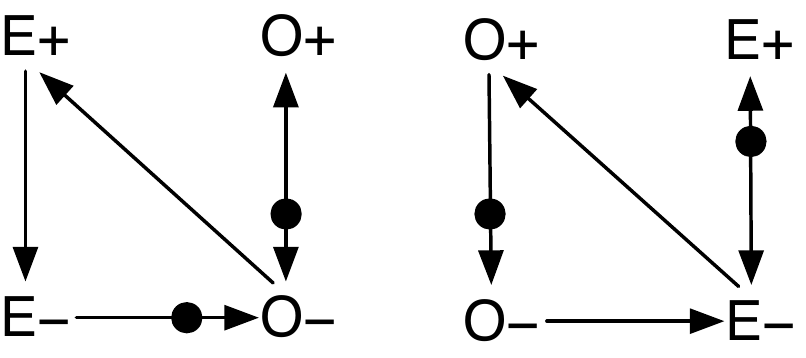}
    \end{center}
    \caption{Rise-Decoupled Protocol}
    \label{fig:rd}
    \end{subfigure}
    \qquad
    \begin{subfigure}{0.3\textwidth}
    \begin{center}
        \includegraphics[scale=0.54]{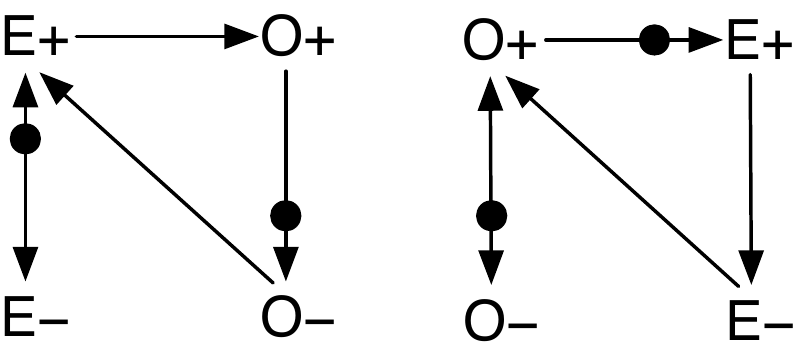}
    \end{center}
    \caption{Fall-Decoupled Protocol}
    \label{fig:fd}
    \end{subfigure}
    \caption{Three marked graph protocols for desynchronization, from \citet{Cortadella2006}. A double arrow $t \leftrightarrow t'$ indicates that there are places in both directions; markings $t \rightarrow t'$ are indicated by dots appearing closer to $t'$ than $t$. For every pair of even-odd (respectively, odd-even) neighbors, the marked graph has the places and initial markings indicated by the diagram on the left (respectively, on the right).}
    \label{fig:protocols}
\end{figure*}

\subsection{Marked graphs}
\label{sec:mg}

\citeauthor{Cortadella2006} characterize the traces allowed by a particular instantiation of controllers as a marked graph, a class of Petri net such that each place has exactly one input and one output transition. 

\begin{definition}[Marked Graphs] Given a set $T$ of transitions, a \emph{marked graph} is a set of places, each with associated input and output transitions, with an initial marking---a function from places to natural numbers. 
\begin{Coq}
Record marked_graph T :=
  { place : T -> T -> Type
  ; init_marking : forall t1 t2, place t1 t2 -> nat }.
Definition marking (M : marked_graph) :=
  forall t1 t2, place M t1 t2 -> nat.
\end{Coq}
\label{def:marked-graphs}
\end{definition}

\Cref{def:marked-graphs} is equivalent to the standard definition of marked graphs as (1) a set of places, (2) a set of transitions, and (3) a flow relation such that each place has unique input and output transitions. However, the type-theoretic \cref{def:marked-graphs} enables easier reasoning in Coq. For example, suppose we have a place $p$ into $t$ in the standard definition. Knowing $t$, there are a finite number of options for $p$, which can be obtained by iterating over the flow relation. \cref{def:marked-graphs} does not require this iteration; Coq automatically identifies the possible options for $p$ by case analysis.

\citeauthor{Cortadella2006} describe marked graph protocols with respect to a circuit's sets of neighboring latches. For each even-odd and odd-even pair they specify the places corresponding to that pair's rises and falls, as illustrated in \cref{fig:protocols}. 
The Coq definition of the desynchronization protocol (\cref{fig:desync}) is shown in \cref{fig:desynchronization-protocol}.

\begin{figure*}
    \begin{center}
    \begin{minipage}[c]{0.55\linewidth}
    \begin{Coq}
Inductive desync_place : event -> event -> Type :=
| fall (l : latch) : desync_place (Rise l) (Fall l)
| rise (l : latch) : desync_place (Fall l) (Rise l)
| rise_fall : forall l l', neighbor c l l' ->
              desync_place (Rise l) (Fall l')
| fall_rise : forall l l', neighbor c l l' ->
              desync_place (Fall l') (Rise l).
Definition desynchronization : marked_graph event :=
  {| place := desync_place
   ; init_marking := fun _ _ p => match p with
     | rise (Even _) => 1
     | fall (Odd _) => 1
     | rise_fall _ _ _ => 1
     | _ => 0
     end |}.
    \end{Coq}
    \end{minipage}
    \end{center}
    \caption{Desynchronization protocol defined in Coq. Given a proof $p$ that the pair $(E,O)$ is in the list of even-odd neighbors, \coqfoot{Even_rise_odd_fall E O p} is a place from $E+$ to $O-$. In the initial marking, the notation \coq{fun x => e} defines a function of type \coq{A -> B} (in this case a \texttt{marking}) by giving a value \coq{e} of type \coq{B} for every argument \coq{x} of type \coq{A}.}
    \label{fig:desynchronization-protocol}
\end{figure*}

\begin{definition}[Enabled Transitions]
A transition is \emph{enabled} in a marking if all of its input places are positive. An enabled transition can \emph{fire} to produce a new marking:
\begin{align*}
    \fire{t}{M}{m}(p) = \begin{cases}
        m(p)-1 &\text{if $p$ leads into $t$} \\
        m(p)+1 &\text{if $p$ leads out of $t$} \\
        m(p) &\text{otherwise}
    \end{cases}
\end{align*}
A marked graph $M$ \emph{evaluates} to the marking $m$ from a tail-list $ls$ of transitions, written $\mg{M}{ls}{m}$, if the following relation holds:
\begin{itemize}
    \item If $ls$ is empty, then ${\mg{M}{ls}{\coqmath{init_marking}(M)}}$.
    \item If $ls=\tcons{t}{ls'}$ such that $\mg{M}{ls'}{m'}$ and $t$ is enabled in $m'$, then $\mg{M}{ls}{\fire{t}{M}{m'}}$.
\end{itemize}
\end{definition}

A crucial property of marked graphs is that they preserve the markings of cycles. We write $p : t_1 \nrightarrow t_2$ for a path spanning from $t_1$ to $t_2$, and $m(p)$ for the sum of the markings of the places along the path.
\begin{lemma} \label{thm:loops}
  If $p$ is a cycle in a marked graph $M$ and $\mg{M}{ls}{m}$, then $m(p)=\coqmath{init_marking}(M)(p)$.
\end{lemma}

This result follows from the following helper lemma:
\begin{lemma}
  If $t$ is enabled in $M$ and $p : t_1 \nrightarrow t_2$, then:
  \begin{align*}
      t=t_1=t_2 &\Longrightarrow \fire{t}{M}{m}(p)=m(p) \\
      t=t_1\neq t_2 &\Longrightarrow \fire{t}{M}{m}(p)=m(p)+1 \\
      t = t_2 \neq t_1 &\Longrightarrow \fire{t}{M}{m}(p)=m(p)-1 \\
      t \neq t_1 \wedge t \neq t_2 &\Longrightarrow \fire{t}{M}{m}(p)=m(p)
  \end{align*}
\end{lemma}

\subsection{Flow equivalence}

\citet{LeGuernic2003} define two circuits to be \emph{flow-equivalent} if and only if (1) they have the same set of latches, and (2) for each latch $l$, the sequence of values latched by $l$ is the same in both circuits. It need not be the case that the two circuits have the same state at any specific point in time, but the projection of the states onto each latch should be the same. This ensures that, even if the timing behaviors of the two circuits differ, their functional behaviors are the same.

In this case, we are interested in comparing the same circuit under two different execution models: synchronous and asynchronous. For the synchronous model, the projection of the values latched by a latch $l$ are given by the function $\synceval{c}{st_0}{i}(l)$. For the asynchronous model, we have to consider the execution of every trace allowed by the asynchronous controller. 

\begin{definition}[Flow Equivalence]
A marked graph $M$ ensures \emph{flow equivalence} of a circuit $c$ with initial state $st_0$ if every trace allowable by the marked graph has the same synchronous and asynchronous execution. That is, let $t$ be a trace compatible with $M$, and let $l$ be a latch, opaque in $t$, such that $\async{c}{st_0}{t}{l}{v}$. Then $v$ is the $i$th value of the synchronous execution of $c$, where $i=\numevents*{l-}{t}$, the number of occurrences of $l-$ in $t$.
\begin{Coq}
Definition flow_equivalence (M : marked_graph event)
  (c : circuit) (st0 : state latch) :=
    forall l t v, (exists m, {M}⊢ t ↓ m) ->
      ⟨c,st0⟩⊢ t ↓ l ↦{Opaque} v ->
      v = sync c st0 (num_events (Fall l) t) l.
\end{Coq}
\end{definition}

\section{Verified Flow-Equivalent Protocols}
\label{sec:proof}
In this section we prove in Coq that the rise-decoupled and fall-decoupled protocols illustrated in \cref{fig:protocols} preserve flow equivalence, and revisit the counterexample of \cref{sec:counter} in the context of the formalization.

\subsection{Rise-decoupled protocol}

The rise-decoupled protocol, also called fully-decoupled~\citep{Fuber1996} and illustrated in \cref{fig:rd}, preserves the invariant that when $l$ is latched for the $i$th time, its predecessors have also been latched the appropriate number of times (either $i$ or $i-1$ times); thus, $l$ correctly latches its $i$th synchronous value. \cref{fig:rd-ind-inv} illustrates several possible interleavings of neighboring latches that satisfy the rise-decoupled protocol. Notice that $B$ may not acquire its correct value by the $i$th occurrence of $B+$, but it will by the $i$th occurrence of $B-$.

\begin{figure}
    \centering
    \includegraphics[width=\columnwidth]{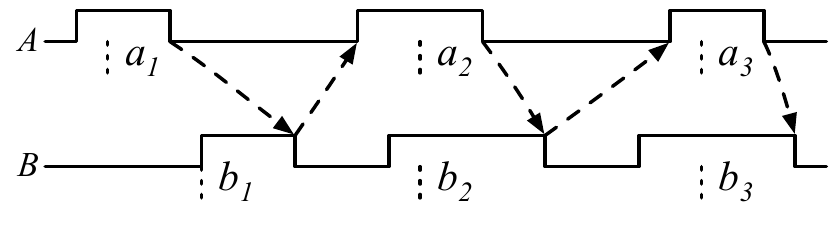}
    \caption{Valid interleavings of the rise-decoupled protocol. Notice that the value $a_i$ is available by the $i$th fall of $B$'s clock.}
    \label{fig:rd-ind-inv}
\end{figure}

Fix a circuit $c$ and an initial state $st_0$. The definition of the \coq{rise_decoupled} marked graph  with respect to $c$ is omitted here for space, but is similar that of the desynchronization protocol shown in \cref{fig:desynchronization-protocol}. In the Coq definition, we add two redundant arcs to ensure that there are always arcs from $l-$ to $l+$ and vice versa. Since there is at most one place between two events $e$ and $e'$, we write $(e \rightarrow e')$ for that place when it exists.

To prove that the rise-decoupled protocol preserves flow equivalence, we first prove several lemmas about the rise-decoupled marked graph.

\begin{lemma} \label{lem:rd-opacity}
  Let $\mg{\coqmath{rise_decoupled}}{t}{m}$. If $l$ has a right neighbor $l'$ such that $m(l- \rightarrow l'-) > 0$, then $l$ is opaque in $t$.
\end{lemma}
\begin{proof}
  By induction on $t$. Using \cref{thm:loops} and some light automation (about 50 lines of Coq tactics that will be re-used for other proofs), this lemma is proved in fewer than 20 lines of Coq proof scripts.
\end{proof}

Next, we prove a lemma relating the number of occurrences of $l-$ in a trace for neighboring latches.
\begin{lemma} \label{lem:rd-num-events}
  Let $\mg{\coqmath{rise_decoupled}}{t}{m}$. If $l-$ is enabled in $m$, then for all left neighbors $l'$ of $l$,
  \begin{align*}
      \numevents*{l'-}{t} = 
      \begin{cases}
        \numevents*{l-}{t} &\text{if}~l~\text{is odd} \\
        1+\numevents*{l-}{t} &\text{if}~l~\text{is even}
      \end{cases}
  \end{align*}
\end{lemma}
\begin{proof}
The induction hypothesis must be strengthened to account for three exhaustive cases: $m(l- \rightarrow l'-) > 0$;
     $m(l'- \rightarrow l+) > 0$; and
     $m(l+ \rightarrow l-) > 0$. These cases are exhaustive because of \cref{thm:loops}: the sum of these three values is exactly equal to $1$.
The proof then proceeds by induction on $t$.
\end{proof}

Finally, we can prove the main result, that \coq{rise_decoupled} satisfies flow equivalence.
\begin{Coq}
Theorem rise_decoupled_flow_equivalence : 
  flow_equivalence rise_decoupled c st0.
\end{Coq}
Recall that we have already fixed the circuit $c$ and initial state $st_0$ and that the definition of the \coq{rise_decoupled} marked graph depends on $c$.

Unfolding the definition of flow equivalence, we can write out the statement of the theorem in English:

\begin{theorem} \label{thm:rd}
  If $\async{c}{st_0}{t}{l}{v}$ such that $l$ is opaque in $t$ and $\mg{\coqmath{rise_decoupled}}{t}{m}$, then $v = \synceval{c}{st_0}{\numevents*{l-}{t}}(l)$.
\end{theorem}
\begin{proof}
  By induction on the asynchronous evaluation judgment $\async{c}{st_0}{t}{l}{v}$. The structure of that definition leads to an induction principle with four cases.
  \begin{enumerate}[wide,labelwidth=!,labelindent=0pt]
      \item The first case, where $l$ is transparent in $t$, is vacuous.
      \item If $t$ is the empty trace it suffices to consider only even latches, for which the result is immediate.
      \item If $t = \tcons{e}{t'}$ such that $e \neq l-$, the result follows immediately from the induction hypothesis.
      
      \item Finally, the only non-trivial case is when $t = \tcons{l-}{t'}$. Since $l-$ is enabled, its left neighbors are all opaque and so have already acquired their correct values.
      
      In this case, $v=\nextstate{c}{st}{l}$ where, for all left neighbors $l'$ of $l$, $\async{c}{st_0}{t'}{l'}{st(l')}$. Since $l-$ is enabled in $\mg{\coqmath{rise_decoupled}}{t'}{m'}$, $l'$ is opaque in $t'$ (\cref{lem:rd-opacity}). Thus the induction hypothesis applies: 
      \[ st(l') = \synceval{c}{st_0}{\numevents*{l'-}{t'}}(l'). \]
      
      Unfolding definitions, \[\synceval{c}{st_0}{\numevents*{l-}{t}}(l) = \nextstate{c}{st'}{l}\]
      where, for all left neighbors $l'$ of $l$,
      \begin{align*}
          st'(l') = 
          \begin{cases}
            \synceval{c}{st_0}{1+i}(l')
                &\text{if}~l~\text{is odd} \\
            \synceval{c}{st_0}{i}(l') 
                &\text{if}~l~\text{is even}
          \end{cases}
      \end{align*}
      and where $i=\numevents*{l-}{t'}$.
      The fact that $st(l')=st'(l')$ follows from  \cref{lem:rd-num-events}. \qedhere
  \end{enumerate}
\end{proof}

Note that we could have performed induction on the trace $t$ here, but that would not be sufficient for the fall-decoupled proof, which requires us to reason about both transparent and opaque latches.

\subsection{Fall-decoupled protocol}

The rise-decoupled protocol of the previous section allows the rises of clocks to interleave arbitrarily as long as the falls of those clocks obey the inductive invariant. In the fall-decoupled protocol (\cref{fig:rd}), the situation is reversed---a clock may fall either before or after its predecessors' clocks fall, as long as its rise occurs after its predecessors' clocks rise, as illustrated in \cref{fig:fd-ind-inv}. Under the zero-delay assumption, the model preserves the invariant that each latch will obtain its correct value as soon as it goes transparent, and that value will be stable until it goes opaque again.

\begin{figure}
    \centering
    \includegraphics[width=\columnwidth]{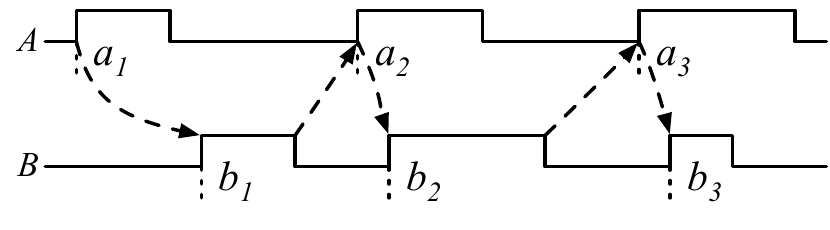}
    \caption{Timing diagram illustrating valid interleavings of the fall-decoupled protocol. Notice that the value $a_i$ is always available by the rise of $B$'s clock.}
    \label{fig:fd-ind-inv}
\end{figure}

For the proof of flow equivalence, we build on the same structure of lemmas as for the rise-decoupled protocols, proving variants of \cref{lem:rd-opacity,lem:rd-num-events} adapted to the fall-decoupled protocol. For example:

\begin{lemma} \label{lem:fd-num-events}
  Let $\mg{\coqmath{fall_decoupled}}{t}{m}$.
  If $l$ is opaque in $t$, then
  \begin{align*}
      \numevents*{l-}{t} = \begin{cases}
        1+\numevents*{l+}{t} &\text{if}~l~\text{is odd} \\
        \numevents*{l+}{t} &\text{if}~l~\text{is even.} \\
      \end{cases}
  \end{align*}
  If $l$ is transparent in $t$, then
  for all left neighbors $l'$ of $l$:
  \begin{align*}
      \numevents*{l+}{t} = \begin{cases}
        \numevents*{l'+}{t} &\text{if}~l~\text{is odd} \\
        1+\numevents*{l'+}{t} &\text{if}~l~\text{is even.} \\
      \end{cases}
  \end{align*}
\end{lemma}

In the rise-decoupled protocol, the proof of flow equivalence relies on the fact that, whenever the event $l-$ is enabled, all of $l$'s left neighbors are opaque, so are guaranteed to have their correct values. This is not the case in the fall-decoupled protocol, where $l$'s left neighbors might be either transparent or opaque. Therefore, it is necessary to strengthen the statement of flow equivalence to account for the values of both opaque and transparent latches.

\begin{theorem} \label{thm:fd-strong}
  If $\mg{\coqmath{fall_decoupled}}{t}{v}$
  and $\async{c}{st_0}{t}{l}{v}$, then $v = \synceval{c}{st_0}{i}(l)$, where
  \begin{align*}
      i = \begin{cases}
        1+\numevents*{l+}{t} &\text{if}~l~\text{is odd} \\
        \numevents*{l+}{t} &\text{if}~l~\text{is even}.
      \end{cases}
  \end{align*}
\end{theorem}
\begin{proof}
  By induction on the relation $\async{c}{st_0}{t}{l}{v}$.
  \begin{enumerate}[wide,labelwidth=!,labelindent=0pt]
  
  \item First, if $l$ is transparent in $t$, then $v$ is equal to $\nextstate{c}{st}{l}$ such that, for all left neighbors $l'$ of $l$, $\async{c}{st_0}{t}{l'}{st(l')}$.
  
  Notice that $i > 0$: for $l$ odd this follows immediately from the definition; for $l$ even there is at least one occurrence of $l+$ in $t$ since even latches are initially opaque. Thus,
  $
      \synceval{c}{st_0}{i}(l) = \nextstate{c}{st'}{l}
  $
  where $st' = \synceval{c}{st_0}{\numevents*{l+}{t}}$.
  
  It suffices to show that for all left neighbors $l'$, 
  $st(l')=st'(l')$.
  By induction, $st(l') = \synceval{c}{st_0}{i'}(l')$ where
  \begin{align*}
      i' = \begin{cases}
        1+\numevents*{l'+}{t} &\text{if}~l'~\text{is odd} \\
        \numevents*{l'+}{t} &\text{if}~l'~\text{is even.}
      \end{cases}
  \end{align*}
 By \cref{lem:fd-num-events} we know that $\numevents*{l+}{t}=i'$, which completes the proof.
  
  \item If $l$ is opaque in the empty trace, it must be odd, and the result is trivial.
  
  \item If $l$ is opaque in $t=\tcons{e}{t'}$ where $e \neq l-$, the result is straightforward from the induction hypothesis on $t'$.
  
  \item Finally, if $t=\tcons{l-}{t'}$, then $v = \nextstate{c}{st}{l}$ where $\async{c}{st_0}{t'}{l'}{st(l')}$ for all left neighbors $l'$ of $l$. Furthermore, $\synceval{c}{st_0}{i}(l)$ is equal to $\nextstate{c}{st'}{l}$ where $st'=\synceval{c}{st_0}{\numevents*{l+}{t'}}$. So it suffices to show that, for all left neighbors $l'$, $st(l')=st'(l')$. By the induction hypothesis $st(l') = \synceval{c}{st_0}{i'}(l')$, where
  \begin{align*}
      i' = \begin{cases}
        1+\numevents*{l'+}{t'} &\text{if}~l'~\text{is odd} \\
        \numevents*{l'+}{t'} &\text{if}~l'~\text{is even}.
      \end{cases}
  \end{align*}
  By \cref{lem:fd-num-events} we know $i' = \numevents*{l+}{t'}$, as required.
  \qedhere
  \end{enumerate}
\end{proof}

Finally, we must show that \cref{thm:fd-strong} implies flow equivalence.

\begin{corollary}
  Whenever $\async{c}{st_0}{t}{l}{v}$ such that $l$ is opaque in $t$ and $\mg{\coqmath{fall_decoupled}}{t}{m}$, then $v = \synceval{c}{st_0}{i}(l)$ where $i = \numevents*{l-}{t}$.
\end{corollary}
\begin{proof}
  It suffices to show that
  \begin{align*}
      \numevents*{l-}{t} = 
      \begin{cases}
        1+\numevents*{l+}{t} &\text{if}~l~\text{is odd} \\
        \numevents*{l+}{t} &\text{if}~l~\text{is even,}
      \end{cases}
  \end{align*}
  which follows from \cref{lem:fd-num-events}.
\end{proof}

\subsection{Less concurrent protocols}

\citeauthor{Cortadella2006} also discuss three other protocols that are even less concurrent than the rise-decoupled and fall-decoupled protocols, and claim they are flow equivalent since every trace they accept is also accepted by a known flow-equivalent protocol. Because they are all three subsumed by the rise- and fall-decoupled protocols, this is still true in our case. The formal proofs of this result depend only on the proof that one marked graph refines another. As an exercise, we proved this fact for the semi-decoupled protocol in our Coq repository.

\subsection{Counterexample, revisited}

Finally, we revisit the counterexample to desynchronization described in  \cref{sec:counter} in the formal setting, providing a concrete circuit and a trace compatible with desynchronization such that, after executing the trace, there is a latch whose value is not equal to $\synceval{c}{st_0}{\numevents*{l-}{t}}(l)$. The proof of flow equivalence fails in this case because neither of the inductive invariants described in the rise-decoupled or fall-decoupled protocols are satisfied---the desynchronization protocol allows latches to go opaque before all of their left neighbors have gone opaque, which can result in an entire value being dropped, as in \cref{fig:counterexample}. 

\begin{theorem} There exists a circuit $c$ for which the desynchronization protocol does not ensure flow equivalence.
\end{theorem}

\begin{proof}
\cref{lst:counterexample} shows a three-stage pipeline, similar to the one in \cref{sec:counter} but with two additional latches $\SRC$ and $\SNK$ representing the input and output environments respectively. In the initial state, all latches have value $X$, and with each pulse, the source environment $\SRC$ sends integers of increasing value to $A$. Each pipeline stage increments its input by one (denoted \coq{inc_value v}), 
and the output environment $\SNK$ consumes the output of $C$.
\begin{figure}
\begin{Coq}
Program Definition c : circuit even odd :=
  {| even_odd_neighbors := [ (A,SRC); (A,B); (C,SNK) ]
   ; odd_even_neighbors := [ (SRC,A); (B,C); (SNK,C) ]
   ; next_state_e := fun E => match E with
        | A => fun st => st (SRC)
        | C => fun st => inc_value (st B)
        end
   ; next_state_o := fun O => match O with
        | SRC => fun st => inc_value (st A)
        | B   => fun st => inc_value (st A)
        | SNK => fun _ => Num 0
        end |}.
\end{Coq}
    \caption{A circuit subject to the desynchronization counterexample.}
    \label{lst:counterexample}
\end{figure}

We take a prefix of the trace illustrated in \cref{fig:counterexample} up until the second fall of $C$'s clock. It is at this point that flow equivalence is violated---the second value latched by $C$ is still $\synceval{c}{st_0}{1}(C)$.
\begin{align*}
    t^c = [&\SNK-, C+, B-, C-, \SNK+, \SNK-, C+, B+, C-]
\end{align*}

It is easy to check that there exists a marking $m$ such that $\mg{\coqmath{desynchronization}}{t^c}{m}$, and that $\async{c}{st_0}{t^c}{C}{\synceval{c}{st_0}{1}(C)}$. To complete the proof that flow equivalence is violated, it suffices to check that $\synceval{c}{st_0}{1}(C) \neq \synceval{c}{st_0}{2}(C)$.
\qedhere

\end{proof}

\section{Summary and Future Work}
\label{sec:summary}

This paper makes three contributions: we identify a mistake in a hand-written proof, carefully formalize the relevant definitions, and adapt the proof to two variants of the original protocol. This research highlights the benefits of formal theorem proving to verify that a design satisfies desirable properties.


Future work will apply this framework to verify the correctness of gate-level implementations of the controllers by checking that they satisfy the rise- or fall-decoupled protocols.  In addition, we plan to extend our Coq framework to account for more abstract design specifications such as FF-based synchronous designs and CSP-like specifications, as well as more realistic delay models. A long-term goal is to develop a comprehensive assurance framework for desynchronization-based design flows, which would account for the correctness of retimed datapaths and liveness. These extensions may require bridging disparate formal frameworks, including static timing analysis, equivalence, and model checking.

\section*{Acknowledgment}
\footnotesize{
Thanks to Marly Roncken, Ivan Sutherland, Tynan McAuley, Flemming Anderson, and the anonymous reviewers for their helpful feedback.
This material is based upon work supported by the Defense Advanced Research Projects Agency (DARPA) under Contract No. HR0011-19-C-0070. The views, opinions, and/or findings expressed are those of the authors and should not be interpreted as representing the official views or policies of the Department of Defense or the U.S. Government. DARPA Distribution Statement ``A'' (Approved for Public Release, Distribution Unlimited).}


\bibliographystyle{IEEEtranN}
\footnotesize
\bibliography{references}
\normalsize

\end{document}